\documentclass[12pt,reqno,oneside]{amsart}

\usepackage[foot]{amsaddr} 
\usepackage{amsthm}
\usepackage{amssymb}
\usepackage{mathrsfs}
\usepackage{tikz}
\usetikzlibrary{arrows, 
	positioning, 
	calc, 
	decorations.pathreplacing, 
	matrix, 
	patterns, 
	patterns.meta, 
	fadings, 
	math, 
	fit
}
\usepackage[margin=1.1in]{geometry}
\usepackage{xcolor}
\usepackage{enumitem}
\usepackage{bbm}
\usepackage[
	colorlinks,
	linkcolor={black!30!blue},
    citecolor={black!30!blue},
    urlcolor={black!30!blue}
]{hyperref}

\hypersetup{
	pdftitle = {Twenty Dry Martinis for the \\unitary Almost Mathieu Operator},
	pdfauthor = {Christopher Cedzich,
		Long Li},
	pdfsubject = { },
	pdfkeywords = {Quantum walks, CMV matrices, Dry ten Martini problem, reducibility, global theory, unitary almost Mathieu operator, quasiperiodic lattice system, discrete time, quantum simulator}
}

\usepackage{cite}

\newcommand{\bbC}{{\mathbb{C}}}
\newcommand{\bbD}{{\mathbb{D}}}

\newcommand{\bbQ}{{\mathbb{Q}}}
\newcommand{\bbR}{{\mathbb{R}}}
\newcommand{\bbS}{{\mathbb{S}}}
\newcommand{\bbT}{{\mathbb{T}}}

\newcommand{\bbZ}{{\mathbb{Z}}}

\newcommand{\calH}{{\mathcal{H}}}

\newcommand{\SL}{{\mathbb{SL}}}
\newcommand{\SU}{{\mathbb{SU}}}

\renewcommand{\Im}{\operatorname{Im}}

\newcommand{\idty}{\mathbbm 1}
\newcommand{\aubrydual}{{\sharp}}

\newcommand{\DC}{{\mathrm{DC}}}

\allowdisplaybreaks

\newtheorem{theorem}{Theorem}[section]

\newtheorem{coro}[theorem]{Corollary}

\newtheorem{definition}[theorem]{Definition}
\newtheorem{remark}[theorem]{Remark}

\DeclareMathOperator{\rot}{rot}

\numberwithin{equation}{section}

\makeatletter
\def\subsection{\@startsection{subsection}{2}%
	\z@{.5\linespacing\@plus.7\linespacing}{.5\linespacing}%
	{\normalfont\scshape\centering}}
\makeatother

\title[Twenty Dry Martinis for the UAMO]{Twenty Dry Martinis for the \\unitary Almost Mathieu Operator}

\author[C.\ Cedzich]{Christopher Cedzich}
\email{\href{mailto:cedzich@hhu.de}{cedzich@hhu.de}}
\address{Heinrich Heine Universit\"at D\"usseldorf, Universit\"atsstr. 1, 40225 D\"usseldorf, Germany}

\author[L.\ Li]{Long Li}
\email{\href{mailto:ll106@rice.edu}{longli@rice.edu}}
\address{Department of Mathematics, Rice University, Houston, TX 77005, USA}

\begin{document}

\begin{abstract}
We solve the Dry Ten Martini Problem for the unitary almost Mathieu operator with Diophantine frequencies in the non-critical regime.
\end{abstract}

\maketitle

\section{Introduction}
The famous ``Ten Martini Problem'' was initially posed by Kac who promised ten Martinis for the solution of the following problem \cite{Kac_martinis}: Consider the almost Mathieu operator (AMO)
\begin{equation}\label{eq:AMO}
    (H_{\lambda,\Phi,\theta} \psi)(n)=\psi(n+1)+\psi(n-1)+2\lambda \cos2\pi(\theta+n\Phi)\psi(n),
\end{equation}
with coupling constant $\lambda>0$, frequency $\Phi\in\bbR\setminus\bbQ$ and phase $\theta\in\bbT$.
Kac asked ``are all the gaps [in the spectrum of the AMO] there?'' or, in other words, ``is the spectrum of the AMO a Cantor set?''. This problem was made public by Simon \cite[Problem 1]{simonAlmostPeriodicSchrodinger1982}. It was completely solved by Avila-Jitomirskaya only decades later \cite{AJ09}, building on considerable effort with partial progress along the way, compare \cite{bellissardCantorSpectrumAlmost1982,choiGaussPolynomialsRotation1990,Puig2004CMP}. For an in-depth overview over the development we refer the reader to the surveys \cite{marxDynamicsSpectralTheory2017,damanik2024two}.

The ``Dry Ten Martini Problem" originally asks the natural follow up question \cite[Problem 2]{simonAlmostPeriodicSchrodinger1982}: {\it are all gaps in the spectrum of the AMO open?} Over the years, substantial partial progress for various subsets of parameters was made by Choi-Elliott-Yui \cite{choiGaussPolynomialsRotation1990}, Puig \cite{Puig2004CMP}, Avila-Jitomirskaya \cite{AJ2012JEMS} and Avila-You-Zhou \cite{avila2023dry}, yet, the critical case $\lambda=1$ with $\Phi$ a Diophantine number is still open. Of course, one can ask the (Dry) Ten Martini Problem for any model of interest. Recent progress has been made for the extended Harper's model \cite{hanDryTenMartini2017}, Sturmian Hamiltonians \cite{band2024dry,bandReviewWorkRaymond2024,raymond1997constructive} and $C^2$- and $\cos$-type sampling functions \cite{WangZhang2017, FV21}. Research along this line turns out to be very fruitful, for example, the discovery of a robust property \cite{GJY} or the construction of explicit examples \cite{HSY2019,HHSY2024}.

We mention in passing that one might take another step in the direction of abstraction and ask the question, how, for a family of random operators defined over an ergodic dynamical system, the base dynamics determines the topological structure of the almost sure spectrum. Clearly, the Dry Ten Martini Problem is a special case of this ``all gaps open'' problem. As for operators with quasiperiodic base dynamics such as the AMO, the set of possible gap labels is determined by the so-called Gap Labeling Theorem \cite{bellissardGapLabellingTheorems1992,bellissardKtheoryAlgebrasSolid1986,damanikGapLabellingDiscrete2023, DEF2024,johnsonRotationNumberAlmost1982a,GeronimoJohnson1996JDE}. The converse of this is the question whether a label predicted by the Gap Labeling Theorem labels an actual gap, i.e., a gap that is not collapsed or degenerate. This converse problem is in general very difficult since both the base dynamics and the sampling functions can affect the topological structure of the spectrum, compare \cite{ABD12}.

A crucial ingredient in the classical proof of the Dry Ten Martini Problem for the AMO is the \emph{Aubry-André} duality \cite{AJ2012JEMS,Puig2004CMP}, that is, a duality between different elements of the family $\{H_\lambda\}$ with respect to a twisted Fourier transform.
In the general case, i.e., without appealing to this duality, the property that ``all gaps are open'' is expected to hold only in generic sense \cite{ABD09,ABD12,DL24,PUIG_SIMO_2006}. On the other hand, whenever one has this duality available for a concrete operator, one can expect the Dry Ten Martini problem to hold.

We here solve the Dry Ten Martini Problem for a model that goes beyond the class of quasiperiodic Schrödinger operators but for which one nevertheless can prove an André-Aubry duality: the so-called \emph{unitary almost Mathieu operator} (UAMO) is a quantum walk that describes the motion of a single particle with two-dimensional internal degree of freedom on the integers in discrete time steps. The dynamics is described by the alternation of two types of operators: a parametrized shift and a ``coin'' operator that locally acts via a matrix-valued quasi-periodic function. This model was introduced and studied in \cite{CFO23}, see also \cite{fillmanSpectralCharacteristicsUnitary2017,CFGW19}. There, among other results concerning the spectral types in various parameter regimes, the Ten Martini Problem was solved in the critical setting for all irrational frequencies. We here advance this result in two important directions: on the one hand, we ``dry'' it, and on the other, we extend it to the non-critical setting.

\section{The model and results}
In this work, we establish the Dry Ten Martini property for a family of quasi-periodic operators that was introduced in \cite{CFO23} and dubbed \emph{unitary almost-Mathieu operator} (UAMO) due to its striking similarities with the almost-Mathieu operator $H_{\lambda,\Phi,\theta}$ defined in \eqref{eq:AMO}. The UAMO is a split-step quantum walk $W_{\lambda}=S_\lambda Q$ on $\calH=\ell^2(\bbZ)\otimes\bbC^2$ that is defined in terms of two operators: 
$S_\lambda$ is the conditional shift operator
\begin{equation*}
    S_\lambda:\delta^\pm_{n}\mapsto\lambda\delta^\pm_{n\pm1}\pm\lambda'\delta^\mp_n,\qquad\lambda\in[0,1],\:\lambda'=\sqrt{1-\lambda^2},
\end{equation*}
and $Q$ is a ``coin'' operator that acts coordinate-wise via a $2\times2$ unitary $Q_n$. For the UAMO, these local coins are quasi-periodically distributed according to the rule
\begin{equation} \label{eq:coindef}
Q_{\lambda,n}
=
Q_{\lambda,n,\Phi,\theta}
=
Q_{\lambda,2\pi(n\Phi + \theta)},
\quad
n \in \bbZ,
\end{equation}
where $Q_{\lambda,\theta}$ denotes a modified counterclockwise rotation through the angle $\theta$, i.e.,
\begin{equation}\label{eq:rotation_coin}
Q_{\lambda,\theta}
=
\begin{bmatrix}
\lambda\cos(\theta)+i\lambda' & -\lambda\sin(\theta) \\
\lambda\sin(\theta) &  \lambda\cos(\theta)-i\lambda'
\end{bmatrix},
\qquad
\theta \in \bbT:=\bbR/\bbZ,\:\lambda\in[0,1],
\end{equation}
and we abbreviated $\lambda'=\sqrt{1-\lambda^2}$. 

With these two building blocks, the UAMO is defined as
\begin{equation}\label{eq:UAMO_def}
    W_{\lambda_1,\lambda_2,\Phi,\theta}:=S_{\lambda_1}Q_{\lambda_2,\Phi,\theta},
\end{equation}
and we shall abbreviate as $W_{\lambda_1,\lambda_2}$ for fixed $\Phi$ and $\theta$.
As discussed in detail in \cite[Section 3]{CFO23}, $\Phi$ plays the role of a magnetic field in an associated two-dimensional model, and $\theta$ plays the role of a Fourier parameter of the second lattice dimension. We shall nevertheless stick with the standard nomenclature in dynamical systems and call $\Phi$ the \emph{frequency} and $\theta$ the \emph{phase} and refer to $\lambda_1$ and $\lambda_2$ as \emph{coupling constants}, as they determine the ``strength'' of the shift and the coin, respectively.

\medskip
The UAMO displays a metal-insulator phase transition with respect to the coupling constants for almost all $\Phi$ and $\theta$: for $\lambda_1>\lambda_2$ the spectrum of $W_{\lambda_1,\lambda_2}$ is purely absolutely continuous, for $0<\lambda_1=\lambda_2\leq1$ it is purely singular continuous while for $\lambda_1<\lambda_2$ it is pure point with exponentially decaying eigenfunctions\footnote{Note that in certain regimes, the ``almost all'' can be lifted to ``all'', for details see \cite[Theorem 2.2]{CFO23}.}. This follows from a characterization of the associated eigenfunction cocycle which in the nomenclature of Avila's global theory \cite{Avila1} results in the following spectral phase diagram \cite{CFO23}:

\definecolor{myblue}{RGB}{100,100,220}
\definecolor{myred}{RGB}{255,100,100}
\definecolor{mygreen}{RGB}{119,221,119}
\definecolor{myyellow}{RGB}{253,253,150}
\definecolor{myorange}{RGB}{255,200,100}
\def\legendwidth{1.75}
\def\legendhight{0.32}
\def\legendposx{1.7}
\def\legendposy{.5}
\def\circlerad{.8ex}
\def\circleshift{-.7ex}
\newcommand\Square[1]{+(-#1,-#1) rectangle +(#1,#1)}
\def\colorsquare#1{\tikz[baseline=\circleshift]\draw[#1,fill=#1,rotate=45] (0,0) \Square{\circlerad};}

\smallskip
\begin{center}
\begin{minipage}{.4\textwidth}
\begin{description}
\item[\colorsquare{myblue}\ \ Subcritical] $\lambda_1>\lambda_2$
\smallskip
\item[\colorsquare{myred}\ \ Critical] $\lambda_1=\lambda_2$
\smallskip
\item[\colorsquare{myorange}\ \ Supercritical] $\lambda_1 < \lambda_2$
\end{description}
\end{minipage}
\begin{minipage}{.2\textwidth}
	\begin{tikzpicture}
		[
		scale=1.5,
		font=\footnotesize
		]
				
		\draw[thick,black] (-.05,-.05) rectangle +(1.1,1.1);
		
		\foreach \i in {0,1}{
			\draw[align=left] (-.02,{\i}) -- (-.08,{\i});
			\draw[align=left] ({\i},-.02) -- ({\i},-.08);
		}
		
		\draw (-.05,0) node[left, align=left] {$0$};
		\draw (-.05,1) node[left, align=left] {$1$};
		\draw (0,-.05) node[below, align=center] {$0$};
		\draw (1,-.05) node[below, align=center] {$1$};
		
		\draw (0,.5) node[left,align=right] {$\lambda_2$};
		\draw (.5,0) node[below,align=center] {$\lambda_1$};
		
		\path[preaction={fill=white}, pattern=north west lines, pattern color=myblue] (0,0) -- (1,1) -- (1,0) -- cycle;
		\path[preaction={fill=white}, pattern=north east lines, pattern color=myorange] (0,0) -- (1,1) -- (0,1) -- cycle;
		
		\draw[very thick, myred] (0,0) -- (1,1);
	\end{tikzpicture}
\end{minipage}
\end{center}

Another important ingredient in these proofs is a unitary version of André-Aubry duality, which relates $W_{\lambda_1,\lambda_2}$ to its ``dual''
\begin{equation*}
    W_{\lambda_1,\lambda_2}^\aubrydual:=W_{\lambda_2,\lambda_1}^\top,
\end{equation*}
and which immediately implies that the spectra of $W_{\lambda_1,\lambda_2}$ and $W_{\lambda_2,\lambda_1}$ are the same.
\smallskip

The arithmetic properties of $\Phi$ play a crucial role in determining spectral properties of the underlying operator. We call $\Phi$ \emph{Diophantine} if there exist $\kappa>0,\tau>1$ such that
\begin{equation}\label{eq.dioCond}
    \|n \Phi\|_{\bbT}:= \inf_{p\in\bbZ} |n \Phi-p|  \geq\frac\kappa{|n|^{\tau+2}}\quad\forall n\neq0.
\end{equation}
In this case, we write $\Phi\in \DC(\kappa,\tau)$. Moreover, we shall denote the set of all Diophantine frequencies by
\begin{equation}\label{eq:dioph}
    \DC=\bigcup_{\kappa>0,\tau>1}\DC(\kappa,\tau),
\end{equation}
which is known to have full Lebesgue measure as a subset of $\bbT$. 
More generally, we say $\rho$ is \emph{Diophantine with respect to $\Phi$} if there exists some positive constants $\kappa',\tau'$ such that
$$
\|\rho - n\Phi\|_{\bbT} \geq \frac{\kappa'}{|n|^{\tau'}}, \qquad \forall \;  0 \neq n \in \mathbb{Z},
$$
and $\rho$ is \emph{rational with respect to $\Phi$} if $\rho - n\Phi\in\mathbb{Q}$ for some $n \in \mathbb{Z}$ rational.

It turns out that the spectrum of $W_{\lambda_1,\lambda_2,\Phi,\theta}$ is independent of $\theta$ by the minimality of the rigid translation $\theta\to \theta+\Phi, \Phi\in \bbR\setminus\bbQ.$ We denote this common spectrum by $\Sigma_{\lambda_1,\lambda_2,\Phi}$.  Our main result is the following:

\begin{theorem}\label{thm.Main1}
For $\Phi\in \DC$ and $\lambda_1\neq\lambda_2$, all gaps in the spectrum of the unitary almost Mathieu operator $W_{\lambda_1,\lambda_2,\Phi}$ allowed by the Gap Labeling Theorem are open. 
\end{theorem}

\begin{remark}\label{rem:twenty}
    The reason for the ``twenty'' instead of the usual ``ten'' in the title of the paper comes from the observation that we have ``twice as many'' gaps as for the AMO: phenomenologically, this can be seen in Figure \ref{fig:Buttervogel} where we see that the spectrum of the UAMO consists of two copies of the butterfly. The reason for this is the prevalence of more symmetries: the spectrum of the AMO is real and symmetric about $0$ due to the involutive symmetry that multiplies locally by $(-1)^n$ and shifts $\theta\mapsto \theta+1/2$. The UAMO has spectrum on the unit circle. It has the same symmetry as the AMO, which amounts to a reflection about the origin. In addition, the spectrum possesses an axial symmetry about the real axis due to complex conjugation \cite{CFGW19}.
\end{remark}

\begin{figure}
    \centering
    \includegraphics[width=.6\textwidth]{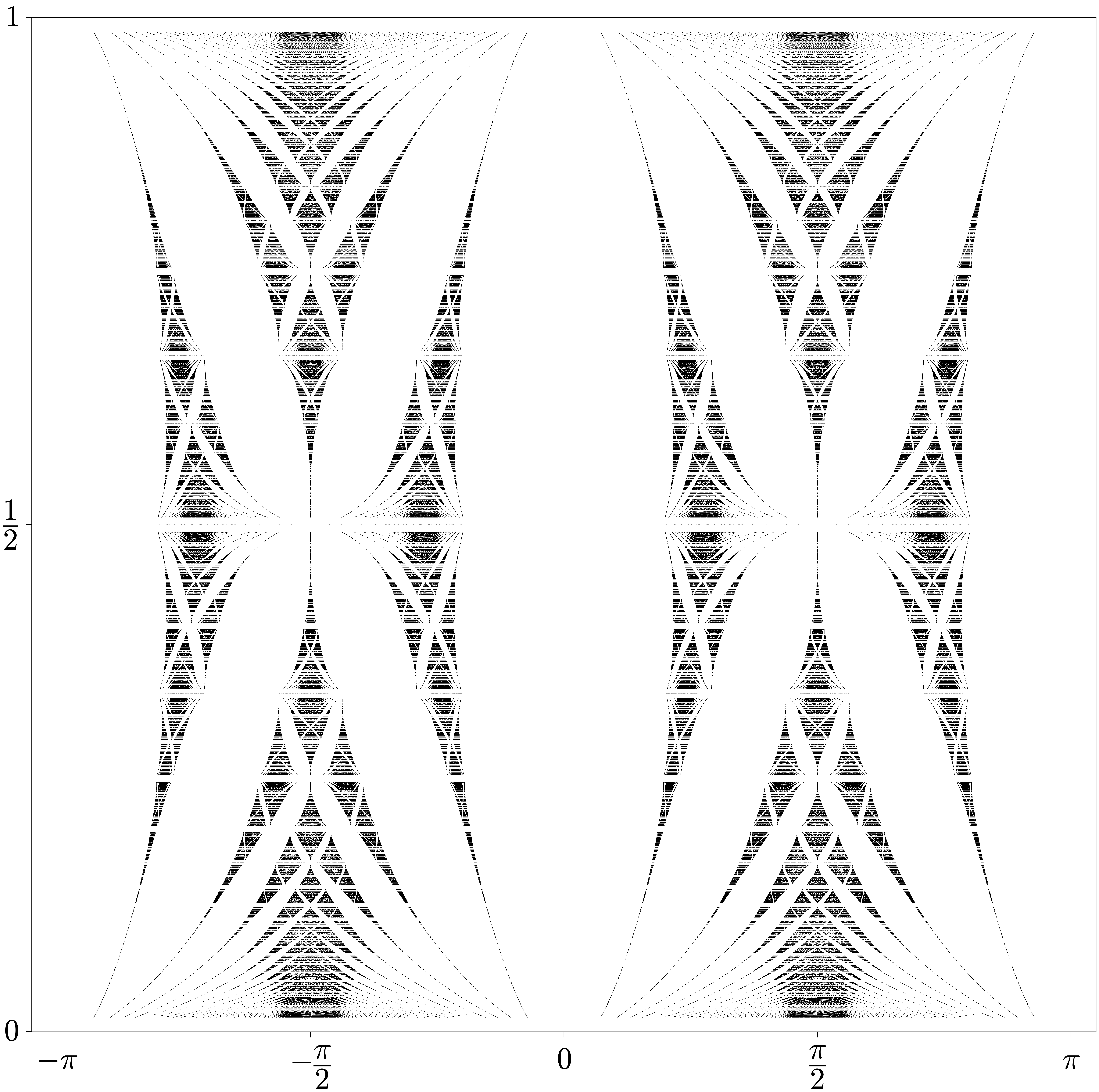}
    \caption{The ``Hofstadter butterfly'' for the UAMO in the subcritical regime with $(\lambda_1,\lambda_2)=(1/\sqrt{2},1/\sqrt{3})$ and denominators up to 70. Clearly, there are two butterflies: for every denominator $q$ there are $2q$ bands instead of just $q$ as for the original butterfly \cite{hof76}. This is rooted in the symmetries of the system: for every $z\in\Sigma_{\lambda_1,\lambda_2,\Phi}$ also $z^*\in\Sigma_{\lambda_1,\lambda_2,\Phi}$ and $-z\in\Sigma_{\lambda_1,\lambda_2,\Phi}$, compare Remark \ref{rem:twenty}.}
    \label{fig:Buttervogel}
\end{figure}

\begin{remark}Let us make some comments about Theorem \ref{thm.Main1}:
\begin{enumerate}[itemsep=1ex]
    \item By \cite[Theorem 2.2 (d)]{CFO23} the Ten Martini property holds in the critical case for all irrational frequencies, that is, for $0<\lambda_1=\lambda_2\leq1$ the spectrum of $W_{\lambda_1,\lambda_2,\Phi}$ is a Cantor set of zero Lebesgue measure for every irrational $\Phi$. The dry version of the problem remains open.
    \item In the non-critical setting, the Dry Ten Martini problem for non-Diophantine $\Phi$ is notoriously hard. For the critical AMO with Liouville frequencies, this was proved in \cite{choiGaussPolynomialsRotation1990} by employing the particular structure and symmetries of the operator. The more involved structure of the UAMO seems to render impossible an adaption of the argument in \cite{choiGaussPolynomialsRotation1990}. We nevertheless expect the DTMP to hold for all frequencies, see \cite{avila2023dry} for a possible attack strategy. 
    \item We combine Avila's almost reducible theorem with Eliasson's reducibility for Diophantine frequencies to obtain the global reducibility except for the critical case. Therefore, this treatment does not cover the Liouville frequencies.
    \item This is the first time a DTMP is showed for (GE)CMV matrices. Some results in Baire category for certain classes of almost periodic extended CMV matrices were previously obtained by \cite{LDZ2022JFA,DL24}. It is an open question whether it holds for other CMV matrices with, for example, subshift or Sturmian Verblunsky coefficients.
\end{enumerate}
\end{remark}

We base our proof on the recent understanding of Anderson localization for Diophantine frequencies obtained in \cite{CFLOZ23} and techniques developed therein:
the Anderson localization for the UAMO in the supercritical setting $\lambda_1<\lambda_2$ proved in \cite{CFO23} is a full measure result. In \cite{CFLOZ23}, an arithmetic version of Anderson localization is proved, albeit for a mosaic model where every other local coin in \eqref{eq:coindef} is trivial. However, the proof of \cite{CFLOZ23} works in a straightforward way for UAMO as well, compare also \cite{FanYang}.

\begin{theorem}\label{thm:AndersonLocal}
Let $\Phi\in \DC(\kappa,\tau)$ and $\lambda_1<\lambda_2$. Then for each ``\,$\Phi$-nonresonant'' $\theta$, i.e., each $\theta$ such that 
\begin{equation*}
    |\sin2\pi(\theta+n\Phi)|<\exp(-|n|^{\frac{1}{2\tau}})
\end{equation*}
does not hold for infinitely many $n$, $W_{\lambda_1,\lambda_2,\Phi,\theta}$ admits Anderson localization.
\end{theorem}

\begin{proof}
In the case $\Phi\in\bbR\setminus\bbQ$ and $\lambda_1<\lambda_2$, according to \cite[Theorem 2.9]{CFO23}, the Lyapunov exponent characterizing the (typical) decay of generalized eigenfunctions is positive:
\begin{equation}\label{eq:positiveLyapunov}
L_{\lambda_1,\lambda_2,\Phi}(z)\geq \log\left[\frac{\lambda_2(1+\lambda_1')}{\lambda_1(1+\lambda_2')}\right]>0,
\end{equation}
with equality if and only if $z\in\Sigma_{\lambda_1,\lambda_2,\Phi}.$ The rest of the proof follows the same outline as the proof of \cite[Theorem 6.3]{CFLOZ23}.
\end{proof}

\begin{remark}
    This result is a full measure result in $\theta$. It is sharp in the sense that it cannot be strengthened to all $\theta$ \cite{cedzich2024absence}.
\end{remark}

We shall also need the following dynamical duality formulation of Autry-André duality for the UAMO, which can be seen as the reverse statement to \cite[Theorem 2.4]{CFO23}. As such, we expect it to be of interest beyond this paper.
\begin{theorem}[Aubry-Andr\'e Duality]\label{thm.reverDual}
Let $\varphi=\varphi^\xi=\left[\varphi^{\xi,+},\varphi^{\xi,-}\right]^\top$, $\xi\in\mathbb{T}$ be a solution to the generalized eigenvalue equation $W^{\aubrydual}_{\lambda_{1},\lambda_{2},\xi,\Phi}\varphi=z\varphi$ which has the following form 
\begin{equation*}
	\begin{bmatrix}\varphi^{\xi,+}_{n}\\\varphi^{\xi,-}_{n}\end{bmatrix}=e^{2\pi in\theta}\begin{bmatrix}\check{\phi}^{+}(\xi+n\Phi)\\\check{\phi}^{-}(\xi+n\Phi)\end{bmatrix}=\frac1{\sqrt2}e^{2\pi in\theta}\begin{bmatrix}\check\psi^+(\xi+n\Phi)+i\check\psi^-(\xi+n\Phi)\\i\check\psi^+(\xi+n\Phi)+\check\psi^-(\xi+n\Phi)\end{bmatrix}.
\end{equation*}
Let \begin{equation}\label{eq:psi_aubr}
	\begin{bmatrix}\check{\psi}^{+}\\\check{\psi}^{-}\end{bmatrix}=\frac1{\sqrt2}\begin{bmatrix}1&-i\\-i&1\end{bmatrix}\begin{bmatrix}\check{\phi}^{+}\\\check{\phi}^{-}\end{bmatrix}
\end{equation}
with $n$-th Fourier coefficients $\psi^{+}_{n}$ and $\psi^{-}_{n}$, respectively. Then $\psi=\left[\psi^{+},\psi^{-}\right]^\top$ solves the eigenvalue equation $W_{\lambda_{1},\lambda_{2},\Phi,\theta}\psi=z\psi$.
\end{theorem}

\section{Preliminaries}

Our proof of Theorem \ref{thm.Main1} utilizes techniques from the theory of one-frequency cocycles of CMV matrices, which we hence review in this section to keep the present treatise as self-contained as possible. We first review the construction of so-called Cantero-Moral-Velázquez matrices (CMV matrices), whose intimate connection with quantum walks on the line was first discussed in \cite{canteroMatrixvaluedSzegoPolynomials2010} and further generalized in \cite{CFO23,CFLOZ23}. We then discuss the dynamics of the transfer matrix cocycle of the UAMO through that of the associated Szeg\H{o} cocycle. This cocycle is the natural one from a point of view of orthogonal polynomials on the unit circle \cite{Simon2005OPUC1,Simon2005OPUC2} and has the advantage that, without further ado, it lies in $\SU(1,1)$.

\subsection{The UAMO as a generalized extended CMV matrix}
Consider the Hilbert space $\ell^2(\bbZ)$ with the standard basis $\{\delta_n : n \in \mathbb{Z}\}$. On $\ell^2(\bbZ)$, we define \emph{generalized extended CMV} (GECMV) matrices $\mathcal{E} = \mathcal{E}(\alpha,\rho)$ by $\mathcal{E}=\mathcal{L}\mathcal{M}$, where $\mathcal{L}=\bigoplus_{n\in\mathbb{Z}}\Theta(\alpha_{2n},\rho_{2n})$ and $\mathcal{M}=\bigoplus_{n\in\mathbb{Z}}\Theta(\alpha_{2n+1},\rho_{2n+1})$ are specified by
\begin{equation}\label{eq:theta_mat}
	\Theta(\alpha,\rho)=\begin{bmatrix}\overline{\alpha}&\rho\\\overline{\rho}&-\alpha\end{bmatrix}
\end{equation}
with Verblunsky pairs
\begin{equation}\label{eq:verblunsky_pair}
	(\alpha,\rho)\in\mathbb{S}^3=\{(z_1,z_2)\in\overline{\mathbb{D}}^2:|z_1|^2+|z_2|^2=1\}.
\end{equation}

Each $\Theta(\alpha_j,\rho_j)$ acts unitarily on the subspace $\ell^2(\{j,j+1\})$, hence the blocks of $\mathcal{L}$ and $\mathcal{M}$ are shifted by one basis element with respect to each other. Hence, $\mathcal E$ has the matrix representation
\begin{equation} \label{eq:gecmv}
	\mathcal E 	= 
        \begin{bmatrix}
		\ddots & \ddots & \ddots & \ddots &&&&  \\
		& \overline{\alpha_0\rho_{-1}} & \boxed{-\overline{\alpha_0}\alpha_{-1}} & \overline{\alpha_1}\rho_0 & \rho_1\rho_0 &&&  \\
		& \overline{\rho_0\rho_{-1}} & -\overline{\rho_0}\alpha_{-1} & {-\overline{\alpha_1}\alpha_0} & -\rho_1 \alpha_0 &&&  \\
		&&  & \overline{\alpha_2\rho_1} & -\overline{\alpha_2}\alpha_1 & \overline{\alpha_3} \rho_2 & \rho_3\rho_2 & \\
		&& & \overline{\rho_2\rho_1} & -\overline{\rho_2}\alpha_1 & -\overline{\alpha_3}\alpha_2 & -\rho_3\alpha_2 &    \\
		&& && \ddots & \ddots & \ddots & \ddots &
	\end{bmatrix},
\end{equation}
where all unspecified matrix elements are zero and we boxed the $\langle\delta_0,\mathcal E\delta_0\rangle$ matrix element. 

``Generalized'' means that compared to the definition of extended CMV matrices given in \cite{Simon2005OPUC1,Simon2005OPUC2,canteroFivediagonalMatricesZeros2003} we admit the $\rho$'s to be complex. By \cite[Proposition 2.12]{CFO23}, resp. \cite[Theorem 2.1]{CFLOZ23}, one has the freedom of changing the phase of the $\rho$'s through a gauge transformation, that is, by conjugating with a diagonal unitary. This freedom has been fruitfully exploited already in several contexts to uncover hidden symmetries of the model \cite{cedzich2024absence,CFLOZ23}.
In particular, every GECMV matrix is gauge-equivalent to a ``standard'' extended CMV matrix with $\rho_n\in[0,1]$ for all $n\in\bbZ$. We shall assume that this gauge transformation has been carried out, i.e., the $\rho$'s are real, and denote the resulting standard extended CMV matrix also by $\mathcal{E}$. This is important, since we want the Szeg\H{o} transfer matrices $S_{n,z}$ defined below in \eqref{eq:szego_normalized} to be in $\mathrm{SU}(1,1)$ and the $\Theta$ in \eqref{eq:theta_mat} to be symmetric.

To study the spectral properties of $\mathcal{E}$, one naturally considers the generalized eigenvalue equation 
\begin{equation*}
    \mathcal{E}u=zu, \qquad z\in\bbC.
\end{equation*}
Solutions to this equation satisfy the following recurrence:
\begin{equation}\label{eq.solIter1}
	\begin{bmatrix}u_{2n+1}\\u_{2n}\end{bmatrix}=A_{n,z}\begin{bmatrix}u_{2n-1}\\u_{2n-2}\end{bmatrix}, \quad n \in \bbZ,
\end{equation}
where the \emph{eigenfunction transfer matrices} $A_{n,z}$ are given by
\begin{equation}\label{eq:GECMV_transmat}
	A_{n,z}=\frac{1}{\rho_{2n}\rho_{2n-1}}\begin{bmatrix}
		z^{-1}+\alpha_{2n}\overline{\alpha}_{2n-1}+\alpha_{2n-1}\overline{\alpha_{2n-2}}+\alpha_{2n}\overline{\alpha_{2n-2}}z& -\overline{\rho_{2n-2}}\alpha_{2n-1}-\overline{\rho_{2n-2}}\alpha_{2n}z\\
		-\rho_{2n}\overline{\alpha_{2n-1}}-\rho_{2n}\overline{\alpha_{2n-2}}z&\rho_{2n}\overline{\rho_{2n-2}}z
	\end{bmatrix},
\end{equation}
for $n \in \bbZ$ and $z \in \bbC \setminus \{0\}$. Since after gauge transforming the $\rho_n$'s in \eqref{eq:GECMV_transmat} are real, by \cite[Lemma 5.3]{CFLOZ23}, we have the following:
 \begin{equation}\label{eq.ConjugatedTransferMatrix}
    A_{n,z}=R_{2n}^{-1}JS^{+}_{n,z}JR_{2n-2},
\end{equation}
where $S^{+}_{n,z}=S_{2n,z}S_{2n-1,z}$ is determined by the \emph{normalized Szeg\H{o} transfer matrices}
\begin{equation}\label{eq:szego_normalized}
    S_{n,z}=\frac{z^{-\frac{1}{2}}}{\rho_{n}}\begin{bmatrix}z&-\overline{\alpha_{n}}\\-\alpha_{n}z&1\end{bmatrix} \in \SU(1,1),
\end{equation}
and
\begin{equation}\label{eq.Rn}
    R_{n}=\begin{bmatrix}1&0\\-\overline{\alpha_{n}}& \rho_{n} \end{bmatrix},\qquad J=\begin{bmatrix}0&1\\1&0\end{bmatrix}.
\end{equation}
\medskip

Consider the generalized eigenvalue equation of the transposed extended CMV matrix:
\begin{equation}\label{eq.dualEigen}
\mathcal{E}^\top v=zv.
\end{equation}
Since $\mathcal{L}$ and $\mathcal{M}$ are symmetric when $\rho\in \bbR$, one has that $\mathcal{E}^\top =\mathcal{M}\mathcal{L}$, and \eqref{eq.dualEigen} becomes $\mathcal{M}\mathcal{L}v=zv$. Applying $\mathcal L$ to both sides, we find that for any generalized eigenfunction $v$ of $\mathcal E^\top$, $u=\mathcal L v$ is a generalized eigenfunction for $\mathcal E$. By definition of $\mathcal L$ we have that
\begin{equation*}
    \begin{bmatrix}u_{2n+1}\\u_{2n}\end{bmatrix}=J\Theta_{2n}J\begin{bmatrix}v_{2n+1}\\v_{2n}\end{bmatrix}=\begin{bmatrix}-\alpha_{2n}&\rho_{2n}\\\rho_{2n}&\overline{\alpha_{2n}}\end{bmatrix}\begin{bmatrix}v_{2n+1}\\v_{2n}\end{bmatrix}.
\end{equation*}
It thus follows from \eqref{eq.solIter1} that solutions $v\in\ell^\infty(\bbZ)$ to \eqref{eq.dualEigen} satisfy the following recurrence relation:
\begin{equation}\label{eq.transposeIter}
\begin{bmatrix}
v_{2n+1}\\v_{2n}
\end{bmatrix}=P_{2n}^{-1}A_{n,z}P_{2n-2}\begin{bmatrix}v_{2n-1}\\v_{2n-2}\end{bmatrix},
\end{equation}
where $P_{n}=J\Theta(\alpha_n,\rho_n)J$ with $J$ as in \eqref{eq.Rn}. This again allows us to easily deduce the Szeg\H{o} transfer matrices for $\mathcal E^\top$ via \eqref{eq.ConjugatedTransferMatrix}.
\medskip

As detailed in \cite[Section 2.3]{CFO23}, the UAMO defined in \eqref{eq:UAMO_def} is a GECMV with dynamically defined Verblunksy coefficients
\begin{equation}\label{eq.uamo}
\begin{aligned}
	\alpha_{2n-1} &= \lambda_2 \sin(2\pi(\theta+n\Phi)),  \qquad\qquad& \alpha_{2n} &= \lambda_1',\\  
    \rho_{2n-1} &=  \lambda_2 \cos(2\pi(\theta+n\Phi)) - i \lambda_2', 	 & \rho_{2n} &= \lambda_1,
    \end{aligned}
\end{equation}
where, as above, $\lambda_i\in[0,1]$ and $\lambda_i'=\sqrt{1-\lambda_i^2}$ for $i=1,2$. Applying the gauge transformation from \cite{CFO23,CFLOZ23} yields an extended CMV matrix with Verblunsky coefficients as above except
\begin{equation}\label{eq.uamo_rho_gauged}
    \rho_{2n-1} =  \left(1-\lambda_2^2 \sin^2(2\pi(\theta+n\Phi))\right)^{1/2}.
\end{equation}

Plugging in these Verblunsky coefficients, $R_{2n}$ defined in \eqref{eq.Rn} is constant. Moreover, $A_{n,z}$ in \eqref{eq:GECMV_transmat} defines a quasiperiodic cocycle $(\Phi,A_z(\cdot))$ which is is $JR$-conjugate to the quasiperiodic two-step combined Szeg\H{o}-cocycle $(\Phi,S_z^{+}(\cdot))$, with $J$ and $R$ given in \eqref{eq.Rn} and
\begin{equation}\label{eq.Sz}
S_z^{+}(\theta)=\frac{1}{\lambda_1\sqrt{1-\lambda_2^2 \sin^2(2\pi \theta)}}\begin{bmatrix}z+\lambda_1'\lambda_2\sin(2\pi \theta)&-\lambda_1'z^{-1}-\lambda_2\sin(2\pi \theta)\\-\lambda_1'z-\lambda_2\sin(2\pi \theta)&z^{-1}+\lambda_1'\lambda_2\sin(2\pi \theta)\end{bmatrix}\in\SU(1,1).
\end{equation}
Similarly, we get the reduced Szeg\H{o} cocycle for the Aubry-dual operator $W^{\aubrydual}_{\lambda_1,\lambda_2,\Phi,\theta}$ from \eqref{eq.transposeIter} and \eqref{eq.ConjugatedTransferMatrix} as
\begin{equation}\label{eq.reverCocy}
S^{\aubrydual}_{z}(\theta)=\frac{1}{\lambda_2\sqrt{1-\lambda_1^2\sin^2(2\pi \theta)}}\begin{bmatrix}z+\lambda_1\lambda_2'\sin(2\pi\theta)&-\lambda_1\sin(2\pi \theta)-\lambda_2'z^{-1}\\-\lambda_1\sin(2\pi \theta)-\lambda_2'z&z^{-1}+\lambda_1\lambda_2'\sin(2\pi \theta)\end{bmatrix}.
\end{equation}
Note that this is just \eqref{eq.Sz} with $\lambda_1$ and $\lambda_2$ exchanged. In the rest of the paper, we will use simply $(\Phi,A_z)$,  $(\Phi,S^+_z)$ and $(\Phi,S^\aubrydual_z)$ to denote these cocycles when they do not cause confusion.

\subsection{Cocycle Dynamics}\label{sec:cocycle_dynamics}

A crucial ingredient to our proof is the behaviour of the quasi-periodic cocycles associated with the UAMO. Let $\Phi\in\bbR\setminus\bbQ$ and $A\in C(\bbT,\SU(1,1))$. This pair $(\Phi,A)$ defines a linear skew-product $(\theta,v)\mapsto(\theta+\Phi,A(\theta)v)$ and is called a $\SU(1,1)$-valued {\it cocycle}. Its iterates are defined as $(n\Phi,A^n)$ where
\begin{equation*}
A^n(\theta)=\begin{cases}
A((n-1)\Phi+\theta)\cdot A((n-2)\Phi+\theta)\cdots A(\theta), & n> 0,\\
A^{-1}(n\Phi+\theta)\cdot A^{-1}((n-1)\Phi+\theta)\cdots A^{-1}(\Phi+\theta), & n<0.
\end{cases}
\end{equation*}
We denote $A^0=\idty$, the identity matrix by usual convention. The motivation for considering $\SU(1,1)$-cocycles comes from the fact that we are interested in cocycles induced by the Szeg\H{o} transfer matrices from \eqref{eq:szego_normalized}. We note that concepts from the theory of $\SL(2,\bbR)$-cocycles carry over directly due to the isomorphism
\begin{equation}\label{eq:M}
    M^{-1}\SU(1,1)M=\SL(2,\bbR),
\end{equation}
where $M$ is the constant unitary matrix
\begin{equation*}
    M=\frac{1}{1+i}\begin{bmatrix}1&-i\\1&i\end{bmatrix}.
\end{equation*}

We adopt the following notions of Avila \cite{Avila1}: we say that $(\Phi,A)$ is \emph{uniformly hyperbolic} if there are $c,C>0$ such that $\|A^n(\theta)\|\geq Ce^{c|n|}$ for all $n\in\bbZ$. Moreover, if $A:X\to \mathrm{SU}(1,1)$ is analytic with an analytic extension to a strip $\{\theta+i\epsilon:|\epsilon|<\delta\}$,
\begin{definition}\label{def.global} If $(\Phi,A)$ is not uniformly hyperbolic, it is said to be
\begin{enumerate}[itemsep=1ex]
    \item \emph{Supercritical}, if $\sup_{\theta\in\bbT}\Vert A^n(\theta)\Vert$ grows exponentially.
    \item \emph{Subcritical}, if there exists a uniform sub-exponential upper bound on the growth of $\Vert A^n(\xi)\Vert$ through some band $|\Im \xi|<\tilde\epsilon$.
    \item \emph{Critical} otherwise.
\end{enumerate}
\end{definition}

This cocycle characterization can be employed to localize the spectrum of quasiperiodic CMV matrices \cite{DFLY2016DCDS}:
\begin{theorem}
    Let $\mathcal E_\theta$, $\theta\in\bbT$, be a quasiperiodic CMV matrix. Then, $z\in\bbC$ is in the spectrum of $\mathcal{E}_\theta$ if and only if the associated Szeg\H{o} cocycle $(\Phi,S_z)$ is not uniformly hyperbolic.
\end{theorem}
In fact, this result holds in a more general setting where instead of circle shifts one merely has a dynamical system $(X,T)$ that is minimal. It can be shown in this case that there exists a common set $\Sigma$ as the spectrum of $\mathcal{E}_\theta$ for every $\theta\in\bbT.$

We call two cocycles $(\Phi,A)$ and $(\Phi,B)$ {\it analytically conjugated} if there is an analytic mapping $Z:\bbT\to\mathrm{PSU}(1,1)$ such that 
\begin{equation*}
[Z(\Phi+\theta)]^{-1}A(\theta)Z(\theta)=B(\theta).
\end{equation*}
We say that $(\Phi,A)$ is \emph{(analytically) reducible} if it is analytically conjugated to a cocycle $(\Phi,B)$ with $B$ constant.
\begin{definition}
We call $(\Phi,A)$ \emph{almost reducible} if the closure of its analytic conjugacy class contains a constant, that is, if there exists $\epsilon>0$ and analytic $Z_n:\bbT\to\mathrm{PSU}(1,1)$ with holomorphic extensions to $\{\theta+iy:|y|<\epsilon\}$ such that
\begin{equation*}
    \lim_{n\to\infty}\|Z_n(\cdot+\Phi)A(\cdot)Z_n(\cdot)^{-1}-B\|_\epsilon=0,
\end{equation*}
where $B\in\mathrm{SU}(1,1)$ is constant and $\|A\|_\epsilon=\sup_{|\Im(z)|<\epsilon}\|A(z)\|$.
\end{definition}

The following result of Avila is crucial for our proof: 
\begin{theorem}[Avila \cite{Avila24}]\label{thm:ARC}
Let $\Phi\in \bbR\setminus\bbQ$. If $(\Phi,A)$ is subcritical, then it is almost reducible.
\end{theorem}

Let $(\Phi,A)$ be an $\SU(1,1)$-valued quasiperiodic cocycle that is homotopic to a constant, and identify $\bbS^1\subset\bbR^2\equiv\bbC$ in the usual way. The cocycle $(\Phi,A)$ induces a homeomorphism $F_A:\bbT\times \bbS^1\to \bbT\times \bbS^1$ by $F_A(\theta,v)=(\Phi+\theta,f_A(\theta,v))$ where
\begin{equation*}
    f_A(\theta,v):=\frac{A(\theta)v}{\|A(\theta)v\|}.
\end{equation*}
This map admits a continuous lift $\widetilde F_A:\bbT\times \bbR\to \bbT\times \bbR$, $(\theta,t)\mapsto(\Phi+\theta,\widetilde f_A(\theta,t))$, where $\widetilde f_A:\bbT\times\bbR\to\bbR$ is a lift of $\widetilde f_A$ satisfying $\widetilde f_A(\theta,t+1)=f(\theta,t)+1$ and, if $\pi_2:\bbT\times\bbR\to\bbT\times\bbS^1$ denotes the projection $(\theta,\phi)\mapsto(\theta,e^{2\pi i\phi})$, $\pi_2\circ\widetilde F_A=F_A\circ\pi_2$.

\begin{definition}[Rotation number \cite{Her83}]
Let $(\Phi,A)$ be a $\SU(1,1)$-valued cocycle that is homotopic to a constant. Then, the limit
\begin{equation*}
    \rot(\Phi,A)=\lim_{n\to\infty}\frac{\widetilde f_A^n(\theta,t)-t}{n}
\end{equation*}
exists uniformly and is independent of $(\theta,t)$. It is called the \emph{fibered rotation number} of the cocycle $(\Phi,A)$.
\end{definition}

The regularity of rotation numbers in great generality without referring to spectral theory can be found in \cite{GordetskiKleptsyn2025}. 

We also need the following discrete one-dimensional analog of a result by Eliasson from \cite{eliassonFloquetSolutions1dimensional1992}, compare also Hadj-Amor \cite[Theorem 1]{HA09} resp. \cite[Theorem 11]{Puig_2005}:
\begin{theorem}\label{Eli92}
Let $\delta > 0$, $\Phi \in \DC(\kappa,\tau)$, and $A_{*}\in \mathrm{SU}(1,1)$. Then there is a constant $\epsilon = \epsilon(\gamma,\tau,\delta,\| A_{0}\|)$ such that if $A \in C^{\omega}_{\delta}(\bbT,\mathrm{SU}(1,1))$ is analytic with
$\| A - A_{*} \|_{\delta} \leq \epsilon$ and the rotation number $\rot(\Phi, A)$ is either Diophantine or rational with respect to $\Phi/2$, then $(\Phi,A)$ is reducible.
\end{theorem}

The cocycle corresponding to the generalized eigenvalue equation \eqref{eq.dualEigen} depends on the spectral parameter $z\in\bbC$, compare \eqref{eq:GECMV_transmat}, therefore we write $\rot\equiv\rot(z)$ to make this dependence explicit.
Indeed, \cite{eliassonFloquetSolutions1dimensional1992} and \cite{moserExtensionResultDinaburg1984}, or more recently \cite{Puig2004CMP}, give a detailed characterization of the location of the spectral parameter based on the knowledge of reducibility for Schr\"odinger cocycles. 

The following result is the form we need; we shall provide a proof along the lines of \cite{Puig_2005} for completeness.
\begin{theorem}\label{thm:gapEdgeRd}
Let $\delta>0$, $\Phi\in \DC(\kappa,\tau)$, $z\in\bbC$, and let $A_z\in \mathrm{SU}(1,1)$ be a constant matrix and let $A_z(\theta)$ be given by \eqref{eq:GECMV_transmat} and \eqref{eq.uamo}. There exists a constant $\epsilon=\epsilon (\kappa,\tau,\delta,\Vert A_z\Vert)$ such that if $A_z(\theta)\in C^\omega_\delta(\bbT,\mathrm{SU}(1,1))$ with $\| A_z(\theta)-A_z\|_\delta<\epsilon$, and $z\in\bbC$ locates at an edge of a spectral gap, then there exists $Z\in C^\omega_\delta(\bbT,\mathrm{PSU}(1,1))$ such that 
$$M^{-1}Z(\theta+\Phi)^{-1}A_z(\theta)Z(\theta)M=\begin{bmatrix}1&c\\0&1\end{bmatrix},$$
where $M$ is the matrix in \eqref{eq:M} that induces the isomorphism between $\mathrm{SU}(1,1)$ and $\mathrm{SL}(2,\bbR)$. In particular, if $\{z\}=\rot^{-1}(k\Phi/2)$ for some $k\in\bbZ$, then $c=0.$
\end{theorem}

\begin{proof}
It suffices to consider the Szeg\H{o} cocycles due to \eqref{eq.ConjugatedTransferMatrix}. For the case $c\neq0$ we shall assume that $c>0$, the case $c<0$ is similar. One may refer to \cite[Theorem 2]{BPS03} for the case $c=0$. Suppose that $\{z\}\in\rot^{-1}(k\Phi/2)$. Then by Theorem \ref{thm:ARC} and Theorem \ref{Eli92}, since we are not in the uniformly hyperbolic case, there exists $Z\in C^\omega_\delta(\bbT,\mathrm{PSU}(1,1))$ such that 
\begin{equation}\label{eq:conjEq1}
    M^{-1}Z(\theta+\Phi)^{-1}\frac{1}{\rho}\begin{bmatrix}z^{\frac{1}{2}}&-\overline{\alpha}z^{-\frac{1}{2}}\\-\alpha z^{\frac{1}{2}}&z^{-\frac{1}{2}}\end{bmatrix}Z(\theta)M=\begin{bmatrix}1&c\\0&1\end{bmatrix}.
\end{equation}
Consider the variation of the spectral parameter $z\mapsto z\,e^{2i\zeta}$ where we put a factor $2$ for convenience. Then, computing the right side of \eqref{eq:conjEq1} yields
\begin{equation}\label{eq:conjEqNew}
\begin{bmatrix}
    1&c\\0&1
\end{bmatrix}+\zeta\begin{bmatrix}
    1&c\\0&1
\end{bmatrix}\begin{bmatrix}-\Im\left((\overline{z_1}-\overline{z_2})(z_1+z_2)\right)&|z_1-z_2|^2\\|z_1+z_2|^2&\Im\left((\overline{z_1}-\overline{z_2})(z_1+z_2)\right)\end{bmatrix}+O(\zeta^2),
\end{equation}
where $Z=\begin{bmatrix}z_1&z_2\\\overline{z_2}&\overline{z_1}\end{bmatrix}$ with $|z_1|^2-|z_2|^2=1.$ Taking the trace and averaging yield to first order 
\begin{equation}\label{eq:trace}
2+c\zeta\left[|z_1+z_2|^2\right]
\end{equation}
where $[\cdot]$ means averaging with respect to the Lebesgue measure on $\bbT.$ It follows that if $c\neq 0$, we can pick $\zeta$ sufficiently small such that $c\zeta>0$ and $\eqref{eq:trace}>2$.

Let us show that the system represented by \eqref{eq:conjEqNew} is uniformly hyperbolic for $c\neq 0$ and $c\zeta>0$ with $\zeta$ sufficiently small.
Let $$B_0=\begin{bmatrix}0&c\\0&0\end{bmatrix}, B_1=\begin{bmatrix}-\Im  \overline{(z_1-z_2)}(z_1+z_2)+\frac{c}{2}|z_1+z_2|^2&|z_1-z_2|^2+c \Im \overline{(z_1-z_2)}(z_1+z_2)\\|z_1+z_2|^2& \Im \overline{(z_1-z_2)}(z_1+z_2)- \frac{c}{2}|z_1+z_2|^2\end{bmatrix}.$$
One can verify that $$\eqref{eq:conjEqNew}=\exp(B_0+\zeta B_1+O(\zeta^2)).$$
Denote
\begin{equation*}
    D(\zeta)=[B_0]+\delta [B_1]=\begin{bmatrix}d_1&d_2\\d_3&-d_1\end{bmatrix}
\end{equation*}
and let 
\begin{equation*}
    d(\zeta)=\det(D(\zeta))=-c\zeta\left[|z_1+z_2|^2\right]+O(\zeta^2).
\end{equation*} 
Then for $\zeta$ sufficiently small satisfying $c\zeta>0$, we have $d(\zeta)<0$.
There exists $Q\in \mathrm{SU}(1,1)$ with $\Vert Q\Vert^2=O(\Vert D\Vert/\sqrt{|\zeta|})$ (compare, e.g., \cite[Lemma 4.1]{LDZ2022JFA}) such that 
$$Q^{-1}D(\zeta)Q=\begin{bmatrix}
    \sqrt{-d}&0\\0&-\sqrt{-d}
\end{bmatrix}.$$
Moreover, $$Q^{-1}\exp(B_0+\zeta B_1+O(\zeta^2))Q=\exp(\Delta+O(|\zeta|^{3/2})),$$
which is uniformly hyperbolic for $\zeta$ sufficiently small.
This shows that when $c\neq0$, $z$ is an edge of an open gap, which contradicts the assumption $\{z\}=\rot^{-1}(k\Phi/2)$. 
\end{proof}

\section{Proofs}
\subsection{Proof of Theorem \ref{thm.Main1}}
In this section, we give a short proof of Theorem \ref{thm.Main1}. In the non-critical case $\lambda_1\neq \lambda_2$, either $S^+_{z}$ or $S^{\aubrydual}_{z}$ is subcritical. Since the corresponding operators $W_{\lambda_1,\lambda_2,\Phi,\theta}$ and $W^{\aubrydual}_{\lambda_1,\lambda_2,\Phi,\theta}$ are isospectral by Aubry-André duality, we can start from either side. Since $\Phi\in \DC$ in Theorem \ref{thm.Main1}, subcriticality essentially implies reducibility \cite{Avila24}, and we can adopt Puig's argument \cite{Puig2004CMP}. The conservation of the Wronskian for the second order difference operator indicated by the transfer matrices rules out double point spectrum, i.e., point spectrum with geometric multiplicity two.
To be more specific, if $W^{\aubrydual}_{\lambda_1,\lambda_2,\Phi,\theta}$ is localized, then its eigenvectors decay to zero, leading to vanishing Wronskians, which means it cannot have two linearly independent eigenvectors corresponding to a single eigenvalue.
The idea of the proof in the following context is that the transfer matrix of $W_{\lambda_1,\lambda_2,\Phi,\theta}$ being reducible to the identity violates the simplicity of the point spectrum of the corresponding Aubry-dual operator, which leads to a contradiction.

We begin by characterizing the two-step Szeg\H{o}-cocycle $(\Phi,S^+_z)$ according to the nomenclature of Definition \ref{def.global}. By \cite[Theorem 2.9]{CFO23} and \eqref{eq.ConjugatedTransferMatrix}, we have the following:
\begin{theorem}\label{thm.subcritical}
Let $\Phi\in \bbR\setminus\bbQ$ and $\lambda_1>\lambda_2$. Then $(\Phi,S^+_{z})$ is subcritical for every $z\in\Sigma_{\lambda_1,\lambda_2,\Phi}$.
\end{theorem}
As a direct corollary, we have that
\begin{coro}\label{cor.subcritical}
Let $\Phi\in \bbR\setminus\bbQ$ and $\lambda_1<\lambda_2$. Then $(\Phi,S^{\aubrydual}_z)$ is subcritical for every $z\in \Sigma_{\lambda_1,\lambda_2,\Phi}$.
\end{coro}
The following result is an analogue of Avron-Simon \cite{AS1983}:
\begin{theorem} For any $|z|=1$,
denote by $\rot_{\lambda_1,\lambda_2}(z)$ and $\rot^\aubrydual_{\lambda_1,\lambda_2}(z)$ the rotation numbers of $(\Phi,S^+_z)$ and $(\Phi,S^\aubrydual_z)$ respectively. Then
$$\rot_{\lambda_1,\lambda_2}(z)=\rot^\aubrydual_{\lambda_1,\lambda_2}(z).$$
\end{theorem}
\begin{proof}
According to Theorem 5.1 of \cite{CFO23}, there exists a unitary transformation $U$ such that 
$$U^*W_{\lambda_1,\lambda_2,\Phi}U=W^\top_{\lambda_2,\lambda_1,\Phi},$$
where $W_{\lambda_1,\lambda_2,\Phi}, W^\top_{\lambda_2,\lambda_1,\Phi}$ are the direct integrals of $W_{\lambda_1,\lambda_2,\Phi,\theta}, W^\top_{\lambda_2,\lambda_1,\Phi,\theta}$ respectively. Since the base dynamics $\theta\to\theta+\Phi$ is minimal and uniquely ergodic, the density of states measures of  $W_{\lambda_1,\lambda_2,\Phi,\theta}$ and $W^\top_{\lambda_2,\lambda_1,\Phi,\theta}$ exist and are denoted by $k(\cdot)$ and $k^\aubrydual(\cdot)$ respectively, compare \cite[Chapter 8]{Simon2005OPUC1}. Moreover, they are independent of the specific choice of $\theta$. It follows that $$k(\cdot)=k^\aubrydual(\cdot).$$
By the relation of rotation numbers of Szeg\H{o} cocycles and the density of states measures of its associated CMV matrices, compare \cite[Theorem 8.3.3]{Simon2005OPUC1}, since our cocycle map is obtained by combining two steps,
$$\rot_{\lambda_1,\lambda_2}(e^{i\zeta})=k(\zeta)=k^\aubrydual(\zeta)=\rot^\aubrydual_{\lambda_1,\lambda_2}(e^{i\zeta}),\quad \zeta\in [0,2\pi).$$
\end{proof}
It is well known that the spectrum equals the subset of $\partial \bbD$ where the rotation number is not locally constant, compare \cite{GeronimoJohnson1996JDE}. Therefore, $W_{\lambda_1,\lambda_2,\Phi,\theta}$ and $W^\aubrydual_{\lambda_1,\lambda_2,\Phi,\theta}$ are isospectral. Moreover, $J\subset\partial \bbD$ is a spectral gap of the spectrum of one of them if and only if it is a spectral gap of the other, and the labels agree.

\begin{proof}[Proof of Theorem \ref{thm.Main1}] 
By the discussion above, it is sufficient to consider the case $\lambda_1<\lambda_2$. 
Assume that $z$ is a gap boundary, that is, $2\rot(\Phi,S^{\aubrydual}_z)=k\Phi$ for some $k\in \mathbb{Z}$. 
Then, by Theorem \ref{thm:ARC} and Corollary \ref{cor.subcritical}, $S^\aubrydual_z$ is almost reducible. Since $\Phi$ is Diophantine, by Theorem \ref{Eli92}, there exists $B\in C^{\omega}(\bbT,\mathrm{PSU}(1,1))$ such that 
\begin{equation}\label{eq.conjugte1}
[B(\theta+\Phi)]^{-1}S^{\aubrydual}_{z}(\theta)B(\theta)=A,
\end{equation}
where $A\in\SU(1,1)$ is constant. By Theorem \ref{thm:gapEdgeRd}, $z\in\rot^{-1}(k\Phi/2)$ is unique if and only if $A=\idty$.

Fix $\theta\in\bbT$ and assume that $z$ is a collapsed gap of the spectrum of $W^\top_{\lambda_2,\lambda_1,\Phi,\theta}$ with $\lambda_1<\lambda_2.$ Since the cocycle $(\Phi,S^{\aubrydual}_{z})$ is subcritical by Corollary \ref{cor.subcritical}, \eqref{eq.conjugte1} boils down to 
\begin{equation*}
[B(\theta+\Phi)]^{-1}S^{\aubrydual}_{z}(\theta)B(\theta)=\idty
\end{equation*}
and therefore, by \eqref{eq.ConjugatedTransferMatrix} and \eqref{eq.transposeIter},
\begin{equation*}
[B(\theta+\Phi)]^{-1}(JRP)^{-1}A_{z}(JRP)B(\theta)=\idty.
\end{equation*}
Let $Z(\theta)=JRPB(\theta)$, then 
\begin{equation}\label{eq.blochW}
[Z(\theta+\Phi)]^{-1}A_{z}(\theta)Z(\theta)=\idty.
\end{equation}
Let $Z(\theta)=\left[U(\theta),V(\theta)\right]$, where $U,V\in C^{\omega}(\bbT,\bbC^2)$ are the columns of $Z(\theta)$. It follows from \eqref{eq.blochW} that 
\begin{equation}\label{eq:UVA}
    \left[U(\theta+\Phi),V(\theta+\Phi)\right]=e^{i0}A_z(\theta)\left[U(\theta),V(\theta)\right],
\end{equation}
and thus $\{U(\theta+n\Phi)\}_{n\in\mathbb{\bbZ}}$ and $\{V(\theta+n\Phi)\}_{n\in\bbZ}$ are independent Bloch waves of the generalized eigenvalue equation $W^\top_{\lambda_2,\lambda_1,\Phi,\theta}\psi=z\psi$. By Theorem \ref{thm.reverDual}, these are independent solutions to  $W_{\lambda_1,\lambda_2,\Phi,0}u=zu$. 

Moreover, by \eqref{eq.solIter1}, \eqref{eq.ConjugatedTransferMatrix} and since $P_{2n}$ and $R_{2n}$ are constant for the Verblunsky coefficients of the UAMO specified in \eqref{eq.uamo} and \eqref{eq.uamo_rho_gauged} after gauge transforming, we have $\det A_{n,z}=1$, which implies the constancy of Wronskian.
Note that we smuggled a factor $e^{i0}$ into \eqref{eq:UVA}: since $0$ is $\Phi$-nonresonant, by Theorem \ref{thm:AndersonLocal} the spectrum of $W_{\lambda_1,\lambda_2,\Phi,0}$ for $\lambda_1<\lambda_2$ and $\Phi\in \DC$ is pure point and simple (by constancy of the Wronksian).
This contradicts the assumption that $z$ is unique, hence the gap cannot be collapsed.
\end{proof}
\medskip

\subsection{Proof of Theorem \ref{thm.reverDual}}
	In the AMO setting the analogous result holds by simply Fourier transforming the assumed eigenvalue equation, so let us try the same strategy here. Writing
	\begin{equation}\label{eq:ft}
		\psi_n=\int_{\mathbb T}\frac{dx}{2\pi}e^{-2\pi i nx}\check\psi(x)
	\end{equation}
	we obtain from $W^{\aubrydual}_{\lambda_{1},\lambda_{2},\xi,\Phi}\varphi=W^{\top}_{\lambda_{2},\lambda_{1},\xi,\Phi}\varphi=z\varphi$ and \cite[Lemma 4.2]{CFO23} that
	\begin{align*}
		z\varphi_n^+ & = (\lambda_1\cos(2\pi(n\Phi+\xi))+i\lambda_1')(\lambda_2\varphi_{n+1}^+ +\lambda_2'\varphi_n^-)+\lambda_1\sin(2\pi(n\Phi+\xi))(-\lambda_2'\varphi_n^++\lambda_2\varphi_{n-1}^-)\\
		z\varphi_n^- & = -\lambda_1\sin(2\pi(n\Phi+\xi))(\lambda_2\varphi_{n+1}^++\lambda_2'\varphi_n^-)+(\lambda_1\cos(2\pi(n\Phi+\xi))-i\lambda_1')(-\lambda_2'\varphi_n^++\lambda_2\varphi_{n-1}^-).
	\end{align*}
	Plugging in the concrete form of $\varphi$ from Theorem \ref{thm.reverDual}, Fourier transforming with respect to $x=\xi+n\Phi$ and multiplying by $\exp[-2\pi in\theta]$, the first equation gives
	\begin{align*}
		\int_{\mathbb T}\frac{dx}{2\pi}e^{-2\pi i mx}z\phi^+(x) &= \int_{\mathbb T}\frac{dx}{2\pi}e^{-2\pi i mx}\Big[(\lambda_1\cos(2\pi x)+i\lambda_1')(\lambda_2e^{2\pi i\theta}\phi^+(x+\Phi) +\lambda_2'\phi^-(x))\\
		&\qquad+\lambda_1\sin(2\pi x)(-\lambda_2'\phi^+(x)+\lambda_2e^{-2\pi i\theta}\phi^-(x-\Phi))\Big].
	\end{align*}
	Expanding the trigonometric functions, reorganizing the terms and utilizing \eqref{eq:ft} gives
	\begin{align*}
		z\phi_m^+ &= \frac12\lambda_2\lambda_1\left[e^{2\pi i((m-1)\Phi+\theta)}\phi_{m-1}^+-ie^{-2\pi i((m-1)\Phi+\theta)}\phi_{m-1}^-\right]\\
		&\qquad+\frac12\lambda_2\lambda_1\left[e^{2\pi i((m+1)\Phi+\theta)}\phi_{m+1}^++ie^{-2\pi i((m+1)\Phi+\theta)}\phi_{m+1}^-\right]\\
		&\qquad+\lambda_2\lambda_1'ie^{2\pi i(m\Phi+\theta)}\phi_m^++i\lambda_2'\lambda_1'\phi_m^-\\
		&\qquad+\frac12\lambda_2'\lambda_1\left[\phi_{m-1}^-+i\phi_{m-1}^++\phi_{m+1}^--i\phi_{m+1}^+\right].
	\end{align*}
	Similarly, from Fourier transforming and expanding the trigonometric functions we obtain from the second equation
	\begin{align*}
		z\phi_m^- & = \frac12\lambda_2\lambda_1(-ie^{2\pi i ((m+1)\Phi+\theta)}\phi_{m+1}^++e^{-2\pi i ((m+1)\Phi+\theta)}\phi_{m+1}^-(x))\\
		&\qquad+ \frac12\lambda_2\lambda_1(ie^{2\pi i ((m-1)\Phi+\theta)}\phi_{m-1}^++e^{-2\pi i ((m-1)\Phi+\theta)}\phi_{m-1}^-(x))\\
		&\qquad+\lambda_2\lambda_1'(-ie^{-2\pi i(m\Phi+\theta)}\phi_m^-)+i\lambda_2'\lambda_1'\phi_m^+\\
		&\qquad+\frac12\lambda_2'\lambda_1\left[i\phi_{m-1}^--i\phi_{m+1}^--\phi_{m-1}^+-\phi_{m+1}^+\right]
	\end{align*}
	Linearly combining these expressions as $a\phi^++b\phi^-$ for $(a,b)=(1,-i)$ and $(a,b)=(-i,1)$ we find
	\begin{align*}
		z\psi_m^+&=\lambda_1(\lambda_2\cos(2\pi((m-1)\Phi+\theta))+i\lambda_2')\psi_{m-1}^+-\lambda_2\lambda_1\sin(2\pi((m-1)\Phi+\theta))\psi_{m-1}^-\\
		&\qquad-\lambda_1'(\lambda_2\cos(2\pi(m\Phi+\theta))-i\lambda_2')\psi_m^--\lambda_2\lambda_1'\sin(2\pi(m\Phi+\theta))\psi_m^+\\
		z\psi_m^-&=\lambda_1(\lambda_2\cos(2\pi((m+1)\Phi+\theta))-i\lambda_2')\psi_{m+1}^-+\lambda_2\lambda_1\sin(2\pi((m+1)\Phi+\theta))\psi_{m+1}^+\\
		&\qquad+\lambda_1'(\lambda_2\cos(2\pi(m\Phi+\theta))+i\lambda_2')\psi_m^+-\lambda_2\lambda_1'\sin(2\pi(m\Phi+\theta))\psi_m^-
	\end{align*}
	where we used the definition of $\psi$ in \eqref{eq:psi_aubr}. Comparing with \cite[Lemma 4.1]{CFO23} we see that indeed $W_{\lambda_1,\lambda_2,\Phi,\theta}\psi=z\psi$.
    \hfill\qedsymbol

\section*{Acknowledgements}
The authors thank David Damanik and Jake Fillman for inspiring discussions. L. Li is supported by AMS-Simons Travel Grant 2024-2026.

\bibliographystyle{abbrvArXiv}
	
\bibliography{DryMartini_bib}

\begin{thebibliography}{10}

\bibitem{HA09}
S.~H. Amor.
\newblock H\"older continuity of the rotation number for quasi-periodic
  co-cycles in {$\mathrm{SL}(2,\bbR)$}.
\newblock {\em Commun.\ Math.\ Phys.}, 287:565--588, 2009.

\bibitem{Avila1}
A.~Avila.
\newblock Global theory of one-frequency {S}chr\"{o}dinger operators.
\newblock {\em Acta Math.}, 215:1--54, 2015.

\bibitem{Avila24}
A.~Avila.
\newblock {KAM}, {L}yapunov exponents, and the spectral dichotomy for typical
  one-frequency {S}chrödinger operators.
\newblock 2023.
\newblock  \href{https://arxiv.org/abs/2307.11071}{{\ttfamily
  arXiv:2307.11071}}.

\bibitem{ABD09}
A.~Avila, J.~Bochi, and D.~Damanik.
\newblock Cantor spectrum for {S}chr\"odinger operators with potentials arising
  from generalized skew-shifts.
\newblock {\em Duke Math. J.}, 146(2):253--280, 2009.

\bibitem{ABD12}
A.~Avila, J.~Bochi, and D.~Damanik.
\newblock Opening gaps in the spectrum of strictly ergodic {S}chr\"odinger
  operators.
\newblock {\em J. Eur. Math. Soc. (JEMS)}, 14(1):61--106, 2012.
\newblock  \href{https://arxiv.org/abs/0903.2281}{{\ttfamily arXiv:0903.2281}}.

\bibitem{AJ09}
A.~Avila and S.~Jitomirskaya.
\newblock The ten {M}artini problem.
\newblock {\em Ann. Math.}, 170(1):303--342, 2009.
\newblock  \href{https://arxiv.org/abs/math/0503363}{{\ttfamily
  arXiv:math/0503363}}.

\bibitem{AJ2012JEMS}
A.~Avila and S.~Jitomirskaya.
\newblock Almost localization and almost reducibility.
\newblock {\em J. Eur. Math. Soc.}, 12(1):93--131, 2010.
\newblock  \href{https://arxiv.org/abs/0805.1761}{{\ttfamily arXiv:0805.1761}}.

\bibitem{avila2023dry}
A.~Avila, J.~You, and Q.~Zhou.
\newblock Dry ten {M}artini problem in the non-critical case.
\newblock 2023.
\newblock  \href{https://arxiv.org/abs/2306.16254}{{\ttfamily
  arXiv:2306.16254}}.

\bibitem{AS1983}
J.~Avron and B.~Simon.
\newblock Almost periodic {S}chr\"odinger operators. {II}. {T}he integrated
  density of states.
\newblock {\em Duke Math. J.}, 50(1):369--391, 1983.

\bibitem{bandReviewWorkRaymond2024}
R.~Band, S.~Beckus, B.~Biber, L.~Raymond, and Y.~Thomas.
\newblock A review of a work by {Raymond}: {Sturmian} {Hamiltonians} with a
  large coupling constant -- periodic approximations and gap labels.
\newblock 2024.
\newblock  \href{https://arxiv.org/abs/2409.10920}{{\ttfamily
  arXiv:2409.10920}}.

\bibitem{band2024dry}
R.~Band, S.~Beckus, and R.~Loewy.
\newblock The dry ten {M}artini problem for {S}turmian {H}amiltonians.
\newblock 2024.
\newblock  \href{https://arxiv.org/abs/2402.16703}{{\ttfamily
  arXiv:2402.16703}}.

\bibitem{bellissardKtheoryAlgebrasSolid1986}
J.~Bellissard.
\newblock K-theory of {C}*-{Algebras} in solid state physics.
\newblock In T.~C. Dorlas, N.~M. Hugenholtz, and M.~Winnink, editors, {\em
  Statistical {Mechanics} and {Field} {Theory}: {Mathematical} {Aspects}},
  pages 99--156, Berlin, Heidelberg, 1986. Springer.

\bibitem{bellissardGapLabellingTheorems1992}
J.~Bellissard.
\newblock Gap {Labelling} {Theorems} for {Schrödinger} {Operators}.
\newblock In M.~Waldschmidt, P.~Moussa, J.-M. Luck, and C.~Itzykson, editors,
  {\em From {Number} {Theory} to {Physics}}, pages 538--630. Springer, Berlin,
  Heidelberg, 1992.

\bibitem{bellissardCantorSpectrumAlmost1982}
J.~Bellissard and B.~Simon.
\newblock Cantor spectrum for the almost {Mathieu} equation.
\newblock {\em J. Funct. Anal.}, 48(3):408--419, 1982.

\bibitem{BPS03}
H.~Broer, J.~Puig, and C.~Sim\'o.
\newblock Resonance tongues and instability pockets in the quasi-periodic
  {H}ill-{S}chr\"odinger equation.
\newblock {\em Comm. Math. Phys.}, 241(2-3):467--503, 2003.

\bibitem{canteroMatrixvaluedSzegoPolynomials2010}
M.~J. Cantero, L.~Moral, F.~A. Gr{\"u}nbaum, and L.~Vel{\'a}zquez.
\newblock Matrix-valued {Szeg{\H o}} polynomials and quantum random walks.
\newblock {\em Commun. Pure Appl. Math.}, 63(4):464--507, 2010.
\newblock  \href{https://arxiv.org/abs/0901.2244}{{\ttfamily arXiv:0901.2244}}.

\bibitem{canteroFivediagonalMatricesZeros2003}
M.~J. Cantero, L.~Moral, and L.~Velázquez.
\newblock Five-diagonal matrices and zeros of orthogonal polynomials on the
  unit circle.
\newblock {\em Linear Algebra Appl.}, 362:29--56, 2003.
\newblock  \href{https://arxiv.org/abs/math/0204300}{{\ttfamily
  arXiv:math/0204300}}.

\bibitem{cedzich2024absence}
C.~Cedzich and J.~Fillman.
\newblock Absence of bound states for quantum walks and {CMV} matrices via
  reflections.
\newblock {\em J. Spectr. Theory}, 14(4):1513–1536, 2024.
\newblock  \href{https://arxiv.org/abs/2402.11024}{{\ttfamily
  arXiv:2402.11024}}.

\bibitem{CFGW19}
C.~Cedzich, J.~Fillman, T.~Geib, and A.~H. Werner.
\newblock Singular continuous {C}antor spectrum for magnetic quantum walks.
\newblock {\em Lett. Math. Phys.}, 110:1141--1158, 2020.
\newblock  \href{https://arxiv.org/abs/1908.09924}{{\ttfamily
  arXiv:1908.09924}}.

\bibitem{CFLOZ23}
C.~Cedzich, J.~Fillman, L.~Li, D.~Ong, and Q.~Zhou.
\newblock Exact mobility edges for almost-periodic {CMV} matrices via gauge
  symmetries.
\newblock {\em Int. Math. Res. Notices}, 2023.
\newblock  \href{https://arxiv.org/abs/2307.10909}{{\ttfamily
  arXiv:2307.10909}}.

\bibitem{CFO23}
C.~Cedzich, J.~Fillman, and D.~C. Ong.
\newblock Almost everything about the unitary almost-{M}athieu operator.
\newblock {\em Commun. Math. Phys.}, 403:745--794, 2023.
\newblock  \href{https://arxiv.org/abs/2112.03216}{{\ttfamily
  arXiv:2112.03216}}.

\bibitem{choiGaussPolynomialsRotation1990}
M.-D. Choi, G.~A. Elliott, and N.~Yui.
\newblock Gauss polynomials and the rotation algebra.
\newblock {\em Invent.\ Math.}, 99(1):225--246, 1990.

\bibitem{DEF2024}
D.~Damanik, I.~Emilsd\'ottir, and J.~Fillman.
\newblock Gap labels and asymptotic gap opening for full shifts.
\newblock 2024.
\newblock  \href{https://arxiv.org/abs/2412.13391}{{\ttfamily
  arXiv:2412.13391}}.

\bibitem{damanikGapLabellingDiscrete2023}
D.~Damanik and J.~Fillman.
\newblock Gap {Labelling} for {Discrete} {One}-{Dimensional} {Ergodic}
  {Schrödinger} {Operators}.
\newblock In M.~Brown, F.~Gesztesy, P.~Kurasov, A.~Laptev, B.~Simon, G.~Stolz,
  and I.~Wood, editors, {\em From {Complex} {Analysis} to {Operator} {Theory}:
  {A} {Panorama}: {In} {Memory} of {Sergey} {Naboko}}, pages 341--404. Springer
  International Publishing, Cham, 2023.
\newblock  \href{https://arxiv.org/abs/2203.03696}{{\ttfamily
  arXiv:2203.03696}}.

\bibitem{damanik2024two}
D.~Damanik and J.~Fillman.
\newblock {\em One-Dimensional Ergodic {S}chr{\"o}dinger Operators: II.
  Specific Classes}, volume 249.
\newblock American Mathematical Society, 2024.

\bibitem{DFLY2016DCDS}
D.~Damanik, J.~Fillman, M.~Lukic, and W.~Yessen.
\newblock Characterizations of uniform hyperbolicity and spectra of {CMV}
  matrices.
\newblock {\em Discrete Contin. Dyn. Syst. Ser. S}, 9(4):1009--1023, 2016.
\newblock  \href{https://arxiv.org/abs/1409.6259}{{\ttfamily arXiv:1409.6259}}.

\bibitem{DL24}
D.~Damanik and L.~Li.
\newblock Opening gaps in the spectrum of strictly ergodic {J}acobi and {CMV}
  matrices.
\newblock 2024.
\newblock  \href{https://arxiv.org/abs/2404.03864}{{\ttfamily
  arXiv:2404.03864}}.

\bibitem{eliassonFloquetSolutions1dimensional1992}
L.~H. Eliasson.
\newblock Floquet solutions for the 1-dimensional quasi-periodic {Schrödinger}
  equation.
\newblock {\em Commun. Math. Phys.}, 146(3):447--482, 1992.

\bibitem{fillmanSpectralCharacteristicsUnitary2017}
J.~Fillman, D.~C. Ong, and Z.~Zhang.
\newblock Spectral {Characteristics} of the {Unitary} {Critical}
  {Almost}-{Mathieu} {Operator}.
\newblock {\em Commun. Math. Phys.}, 351(2):525--561, 2017.
\newblock  \href{https://arxiv.org/abs/1512.07641}{{\ttfamily
  arXiv:1512.07641}}.

\bibitem{FV21}
Y.~Forman and T.~VandenBoom.
\newblock Localization and {C}antor spectrum for quasiperiodic discrete
  {S}chr\"odinger operators with asymmetric, smooth, cosine-like sampling
  functions.
\newblock 2021.
\newblock  \href{https://arxiv.org/abs/2107.05461}{{\ttfamily
  arXiv:2107.05461}}.

\bibitem{GJY}
L.~Ge, S.~Jitomirskaya, and J.~You.
\newblock {K}otani theory, {P}uig's argument, and stability of the ten
  {M}artini problem.
\newblock 2023.
\newblock  \href{https://arxiv.org/abs/2308.09321}{{\ttfamily
  arXiv:2308.09321}}.

\bibitem{GeronimoJohnson1996JDE}
J.~S. Geronimo and R.~A. Johnson.
\newblock Rotation number associated with difference equations satisfied by
  polynomials orthogonal on the unit circle.
\newblock {\em J. Differ. Equations}, 132(1):140--178, 1996.

\bibitem{GordetskiKleptsyn2025}
A.~Gorodetski and V.~Kleptsyn.
\newblock Log h\"older continuity of the rotation number, 2024.
\newblock  \href{https://arxiv.org/abs/2410.15462}{{\ttfamily
  arXiv:2410.15462}}.

\bibitem{hanDryTenMartini2017}
R.~Han.
\newblock Dry {Ten} {Martini} problem for the non-self-dual extended
  {Harper}’s model.
\newblock {\em T. Am. Math. Soc.}, 370(1):197--217, 2017.
\newblock  \href{https://arxiv.org/abs/1607.08571}{{\ttfamily
  arXiv:1607.08571}}.

\bibitem{HHSY2024}
J.~He, X.~Hou, Y.~Shan, and J.~You.
\newblock Explicit construction of quasi-periodic analytic {S}chr\"odinger
  operators with cantor spectrum.
\newblock {\em Math. Ann.}, 391(1):179--225, 2025.
\newblock  \href{https://arxiv.org/abs/2312.16434}{{\ttfamily
  arXiv:2312.16434}}.

\bibitem{Her83}
M.~Herman.
\newblock Une m\'ethode pour minorer les exposants de {L}yapounov et quelques
  exemples montrant le caract\`ere local d'un th\'eor\`eme d' {A}rnold et de
  {M}oser sur le tore de dimension 2.
\newblock {\em Comment. Math. Helv.}, 58:453--502, 1983.

\bibitem{hof76}
D.~R. Hofstadter.
\newblock Energy levels and wave functions of {B}loch electrons in rational and
  irrational magnetic fields.
\newblock {\em Phys. Rev. B}, 14:2239--2249, 1976.

\bibitem{HSY2019}
X.~Hou, Y.~Shan, and J.~You.
\newblock Construction of quasiperiodic {S}chr\"odinger operators with {C}antor
  spectrum.
\newblock {\em Ann. Henri Poincar\'e}, 20(11):3563--3601, 2019.

\bibitem{johnsonRotationNumberAlmost1982a}
R.~Johnson and J.~Moser.
\newblock The rotation number for almost periodic potentials.
\newblock {\em Commun. Math. Phys.}, 84(3):403--438, 1982.

\bibitem{Kac_martinis}
M.~Kac.
\newblock public commun. at the 1981 AMS Annual Meeting.

\bibitem{LDZ2022JFA}
L.~Li, D.~Damanik, and Q.~Zhou.
\newblock Cantor spectrum for {CMV} matrices with almost periodic {V}erblunsky
  coefficients.
\newblock {\em J. Funct. Anal.}, 283(12):109709, 2022.
\newblock  \href{https://arxiv.org/abs/2112.14376}{{\ttfamily
  arXiv:2112.14376}}.

\bibitem{marxDynamicsSpectralTheory2017}
C.~A. Marx and S.~Jitomirskaya.
\newblock Dynamics and spectral theory of quasi-periodic {Schrödinger}-type
  operators.
\newblock {\em Ergodic Theory and Dynamical Systems}, 37(8):2353--2393, 2017.
\newblock  \href{https://arxiv.org/abs/1503.05740}{{\ttfamily
  arXiv:1503.05740}}.

\bibitem{moserExtensionResultDinaburg1984}
J.~Moser and J.~Pöschel.
\newblock An extension of a result by {Dinaburg} and {Sinai} on quasi-periodic
  potentials.
\newblock {\em Comment. Math. Helv.}, 59(1):39--85, 1984.

\bibitem{Puig2004CMP}
J.~Puig.
\newblock Cantor spectrum for the almost {M}athieu operator.
\newblock {\em Commun. Math. Phys.}, 244(2):297--309, 2004.
\newblock  \href{https://arxiv.org/abs/math-ph/0309004}{{\ttfamily
  math-ph/0309004}}.

\bibitem{Puig_2005}
J.~Puig.
\newblock A nonperturbative {E}liasson’s reducibility theorem.
\newblock {\em Nonlinearity}, 19(2):355–376, 2005.
\newblock  \href{https://arxiv.org/abs/math/0503356}{{\ttfamily
  arXiv:math/0503356}}.

\bibitem{PUIG_SIMO_2006}
J.~Puig and C.~Sim\'o.
\newblock Analytic families of reducible linear quasi-periodic differential
  equations.
\newblock {\em Ergod. Theor. and Dyn. Syst.}, 26(2):481–524, 2006.

\bibitem{raymond1997constructive}
L.~Raymond.
\newblock A constructive gap labelling for the discrete schr{\"o}dinger
  operator on a quasiperiodic chain.
\newblock {\em preprint}, 1995.

\bibitem{simonAlmostPeriodicSchrodinger1982}
B.~Simon.
\newblock Almost periodic {Schrödinger} operators: {A} {Review}.
\newblock {\em Adv. Appl. Math.}, 3(4):463--490, 1982.

\bibitem{Simon2005OPUC1}
B.~Simon.
\newblock {\em Orthogonal polynomials on the unit circle. {P}art 1}, volume~54
  of {\em American Mathematical Society Colloquium Publications}.
\newblock American Mathematical Society, Providence, RI, 2005.
\newblock Classical theory.

\bibitem{Simon2005OPUC2}
B.~Simon.
\newblock {\em Orthogonal polynomials on the unit circle. {P}art 2}, volume~54
  of {\em American Mathematical Society Colloquium Publications}.
\newblock American Mathematical Society, Providence, RI, 2005.
\newblock Spectral theory.

\bibitem{WangZhang2017}
Y.~Wang and Z.~Zhang.
\newblock Cantor spectrum for a class of {$C^2$} quasiperiodic {S}chr\"odinger
  operators.
\newblock {\em Int. Math. Res. Notices}, (8):2300--2336, 2017.
\newblock  \href{https://arxiv.org/abs/1410.0101}{{\ttfamily arXiv:1410.0101}}.

\bibitem{FanYang}
F.~Yang.
\newblock Anderson localization for the unitary almost {M}athieu operator.
\newblock {\em Nonlinearity}, 37(8):Paper No. 085010, 23, 2024.
\newblock  \href{https://arxiv.org/abs/2201.05779}{{\ttfamily
  arXiv:2201.05779}}.

\end{thebibliography}

\end{document}